\documentclass[11pt]{article}
\usepackage[letterpaper, margin=1in]{geometry}
\usepackage[utf8]{inputenc}
\usepackage{authblk}

\usepackage{amsmath,amssymb,amsthm,mathtools}
\usepackage{microtype}
\usepackage{cite}
\usepackage[colorlinks=true,allcolors=blue]{hyperref}

\newtheorem{theorem}{Theorem}

\newtheorem{lemma}{Lemma}

\theoremstyle{definition}
\newtheorem{problem}[lemma]{Problem}

\theoremstyle{remark}

\newcommand{\norm}[1]{\left\lVert #1\right\rVert}
\newcommand{\ket}[1]{|#1\rangle}
\newcommand{\bra}[1]{\langle #1|}
\newcommand{\braket}[2]{\langle #1|#2\rangle}
\newcommand{\Adv}{\operatorname{Adv}}
\newcommand{\e}{\mathrm e}
\newcommand{\iu}{\mathrm i}

\title{Optimal Joint Lower Bound for Sparse\\ Quantum Linear System Solvers}
\author{Dhrumil Patel\\ Department of Computer Science, Virginia Tech, Alexandria, VA 22305, USA}

\begin{document}
\maketitle

\begin{abstract}
Quantum linear system solvers form one of the central algorithmic primitives in quantum computing, with applications ranging from differential equations and optimization to machine learning. 
Their cost is commonly measured through query complexity, which counts the number of oracle calls needed to access the input matrix. 
In the sparse-access model, this complexity is governed by three parameters: the condition number $\kappa$, the sparsity $s$ of the input matrix, and the target precision $\varepsilon$. 
The dependence on $\kappa$ and $\varepsilon$ is already well understood through the lower bound $\Omega(\kappa\log(1/\varepsilon))$, which matches the best known scaling in these parameters. 
Once the sparsity $s$ is included, however, the expected lower bound has long been conjectured to be $\Omega(\kappa\sqrt{s}\log(1/\varepsilon))$.
Recent work by Mori et al. [Quantum Sci. Tech. 11 035063 (2026)] made an important step towards this goal by establishing the lower bound $\Omega(\kappa\sqrt{s})$ for constant error $\varepsilon$. 
In this work, we complete the picture and prove the full joint lower bound $\Omega(\kappa\sqrt{s}\log(1/\varepsilon))$ in the sparse-access model, thereby establishing the anticipated dependence on all three parameters simultaneously.
\end{abstract}

\section{Introduction}

The quantum linear system (QLS) problem is to prepare a quantum state that represents the solution of a linear system $Ax=b$. 
Given an invertible matrix $A\in\mathbb{C}^{N\times N}$ together with a normalized state $\ket{b}\in\mathbb{C}^N$, the goal is to output a normalized state $\ket{\widetilde{x}}$ that is close to the exact solution state $\ket{x}\coloneqq A^{-1}\ket{b}/\lVert A^{-1}\ket{b}\rVert$. 
More precisely, we require $\lVert\ket{\widetilde{x}}-\ket{x}\rVert\le\varepsilon$, where $\lVert\cdot\rVert$ denotes the standard Euclidean norm. 
In the sparse-access model considered here, the complexity of this task is governed by three parameters: the condition number $\kappa$ of $A$, the sparsity $s$ of $A$, and the target error $\varepsilon$.

The QLS problem was originally introduced by Harrow, Hassidim, and Lloyd~\cite{Harrow2009}, and subsequent work has progressively improved the dependence on both the condition number $\kappa$ and the target error $\varepsilon$. 
In the block-encoding model, the best known algorithms achieve query complexity $O(\kappa\log(1/\varepsilon))$~\cite{Costa2022,Dalzell2024, Cunningham2024}.  
In the sparse-access model, however, the matrix $A$ is not supplied through a block encoding.  
Instead, the algorithm accesses the locations and values of the nonzero entries of $A$ through specific sparse-matrix oracles, which we define formally in~\eqref{eq:position-oracle} and~\eqref{eq:value-oracle}.
This introduces the sparsity $s$ as an additional complexity parameter.  
In this model, Low gave a QLS algorithm using
$O((\kappa\sqrt{s})^{1+o(1)}/\varepsilon^{o(1)})$
queries to the sparse-matrix oracles~\cite{Low2019}.

The corresponding lower bounds have, until now, captured different subsets of these parameters.  
The dependence on $\kappa$ and $\varepsilon$ is captured by the lower bound $\Omega(\kappa\log(1/\varepsilon))$.  
The original proof of this lower bound is attributed to Harrow and Kothari, although their result remains unpublished, and detailed derivations have appeared in later work~\cite{Costa2025,Mori2026}.  
Once the sparsity $s$ is included, the joint lower bound has long been expected to take the form
\begin{equation}\label{eq:expected-lower-bound}
    \Omega(\kappa\sqrt{s}\log(1/\varepsilon)).
\end{equation}
While this specific scaling is often treated as folklore in the QLS literature, a rigorous proof for the joint dependence on all three parameters has been missing.

More recently, the authors of Ref.~\cite{Mori2026} made an important step toward this result by proving a lower bound of $\Omega(\kappa\sqrt{s})$ for instances where the target error $\varepsilon$ is fixed to a constant.
Their result successfully establishes the joint dependence on $\kappa$ and $s$, but it does not capture the logarithmic dependence on the error, $\log(1/\varepsilon)$.  
Therefore, establishing a single lower bound that simultaneously captures the dependence on all three parameters remained an open problem.

Our main result, Theorem~\ref{thm:main}, resolves this open problem and proves~\eqref{eq:expected-lower-bound}.  
It is worth emphasizing in what sense this scaling is already the strongest one suggested by known QLS algorithms.  
If we fix the error $\varepsilon$ to a constant, Low's algorithm achieves a query complexity of $O((\kappa\sqrt{s})^{1+o(1)})$~\cite{Low2019}, which implies that we cannot increase the
polynomial dependence on the sparsity beyond $\sqrt{s}$. 
On the other hand, if we fix the sparsity $s$ to a constant, we can convert sparse access into block-encoding access with only a constant overhead. Once we have a block encoding, we can solve QLS using $O(\kappa\log(1/\varepsilon))$ queries~\cite{Costa2022,Dalzell2024, Cunningham2024}. 
This tells us that neither the polynomial power of $\kappa$ nor the logarithmic dependence on $1/\varepsilon$ can be strictly increased.  
Thus, the lower bound in~\eqref{eq:expected-lower-bound} tightly brings together the strongest parameter-wise dependencies compatible with the best known upper bounds.

Our proof begins with a Boolean query problem that we call \textsc{Position-Parity}.  
The input to this problem is a length-$m$ bit string containing exactly one $1$, and the task is to determine whether the position of this unique $1$ is an even integer or an odd integer.  
We show that solving one instance of the \textsc{Position-Parity} problem requires $\Omega(\sqrt{m})$ quantum queries.  
Later in the proof, we choose $m=\Theta(s)$, which produces the factor $\sqrt{s}$ in our final QLS lower bound.

To obtain the remaining scaling factors $\kappa$ and $\log(1/\varepsilon)$, we combine $n$ independent instances of the \textsc{Position-Parity} problem and reduce the computation of their collective parity (the XOR of the $n$ answers) to the solution of a single sparse linear system. 
We choose the number of instances to be $n=\Theta(\kappa\log(1/\varepsilon))$.
We then apply the XOR lemma (Lemma~\ref{lem:xor}) of Lee and Roland~\cite{Lee2012}, which allows us to preserve the $\Omega(\sqrt{m})$ query hardness of the individual \textsc{Position-Parity} instances under this $n$-fold composition.
Applying the XOR lemma yields a Boolean query lower bound of $\Omega(n\sqrt{m})$, which becomes the desired QLS lower bound $\Omega(\kappa\sqrt{s}\log(1/\varepsilon))$ after substituting $m=\Theta(s)$ and $n=\Theta(\kappa\log(1/\varepsilon))$.

\paragraph{Concurrent work.}
In concurrent and independent work posted after the first version of this paper, Bravo-Prieto et al.~\cite{BravoPrieto2026} establish the same joint lower bound and obtain a matching upper bound of $O(\kappa\sqrt{s}\log(1/\varepsilon))$ in the sparse-access model.
Together, these results establish the optimal quantum query complexity $\Theta(\kappa\sqrt{s}\log(1/\varepsilon))$ for QLS in this model.

\paragraph{Organization.}
The rest of the paper develops our reduction in order.
In Section~\ref{sec:prelim}, we formally introduce the QLS problem, the access models, and the query-complexity tools that we use in the proof, including the XOR lemma of Ref.~\cite{Lee2012}.
In Section~\ref{sec:strategy}, we state our main result and give a proof sketch of the reduction.
We then develop each step of the reduction in detail in Sections~\ref{sec:selector}--\ref{sec:lower}.

\section{Preliminaries and Access Model}\label{sec:prelim}

We first fix the notation and query-complexity tools used throughout the paper.  
For a vector $v$, we write $\norm{v}$ for its Euclidean norm, and for a matrix $A$, we write $\norm{A}$ for its spectral norm and $\norm A_1$ for its trace norm.
We denote the largest and smallest singular values of $A$ by $\sigma_{\max}(A)$ and $\sigma_{\min}(A)$, respectively.
If a matrix $A$ is invertible, we denote its condition number as $\kappa(A) = \sigma_{\max}(A)/\sigma_{\min}(A)$.  
We call a matrix $A$ $s$-sparse if every row and every column of $A$ contains at most $s$ nonzero entries. 
We use the symbol $\circ$ to denote the Hadamard (entrywise) matrix product.  
For a positive integer $n$, we use the standard set shorthand $[n]\coloneqq\{1,\ldots,n\}$, and let $\mathbb Z_m=\{0,1,\ldots,m-1\}$.  
We index bit positions from zero, so an input $x\in\{0,1\}^m$ has bits $x_q$ with $q\in\mathbb Z_m$.
We use~$\oplus$ to denote the bitwise XOR operation (addition modulo 2).
For a Boolean function $f$, we write $Q_\eta(f)$ for the minimum number of quantum queries required to compute $f$ with worst-case error probability at most $\eta$  on every promised input.
All logarithms in this paper are the natural logarithms unless otherwise specified.

We now state the standard formulation of the quantum linear system problem.

\begin{problem}[Quantum linear system (QLS)]\label{prob:qls}
Let $A\in\mathbb{C}^{N\times N}$ be an invertible matrix and let $\ket{b}\in\mathbb{C}^{N}$ be a normalized quantum state.  
We define the exact normalized solution state as
\begin{equation}\label{eq:solution-state}
    \ket{x}
    \coloneqq
    \frac{A^{-1}\ket{b}}
         {\norm{A^{-1}\ket{b}}}.
\end{equation}
Given access to $A$, a procedure for preparing $\ket{b}$, and a target error $\varepsilon>0$, the QLS problem is the task of preparing a quantum state $\ket{\widetilde{x}}$ such that $\norm{\ket{\widetilde{x}}-\ket{x}} \le \varepsilon$.
\end{problem}

To discuss the query complexity of Problem~\ref{prob:qls}, we need to specify how a QLS algorithm accesses the input matrix $A$.
Throughout this paper, we consider the standard sparse-matrix access model.  
Suppose that the matrix $A$ is $s$-sparse.  
For each row $i\in[N]$ and each sparse index $\ell\in[s]$ with $1\le s\le N$, let $\pi_A(i,\ell)\in[N]$
denote the column index stored in the $\ell$-th position of the sparse representation of row $i$.
The set of column indices $\{\pi_A(i,\ell):\ell\in[s]\}$ contains the locations of all nonzero entries located in row $i$ of the matrix $A$.  
If row $i$ contains fewer than $s$ nonzero entries, the remaining sparse indices $\ell$ may point to entries $A_{ij}=0$.

A quantum algorithm queries the sparse representation of $A$ via two quantum oracles.  
The location of the nonzero entries is accessed through the sparse-position oracle, defined as
\begin{equation}\label{eq:position-oracle}
    O_A^{\mathrm{pos}}
    \ket{i,\ell,j}
    =
    \ket{i,\ell,j\oplus \pi_A(i,\ell)},
\end{equation}
where the final register stores a column index in binary.
Thus, given a row $i$ and sparse index $\ell$, the oracle returns the corresponding column index $\pi_A(i,\ell)$.
The numerical value of the corresponding matrix entry is accessed through the sparse-value oracle, defined as
\begin{equation}\label{eq:value-oracle}
    O_A^{\mathrm{val}}\ket{i,j,z}
=
\ket{i,j,z\oplus [A_{ij}]},
\end{equation}
where $[A_{ij}]$ denotes a fixed binary representation of $A_{ij}$.
For the hard instances we construct in this paper, every matrix entry is an exact dyadic rational (a number of the form $a/2^r$ for integers $a$ and $r\ge 0$). 
Thus, this binary representation is exact and introduces no rounding error. 
To this end, the query complexity of a QLS algorithm is defined as the total number of calls made to the oracles $O_A^{\mathrm{pos}}$, $O_A^{\mathrm{val}}$, and their inverses.
Note that these oracles are self-inverse.
The preparation of $\ket{b}$ is available independently and is not counted as a matrix query.

Our proof relates these sparse-matrix queries to queries made in the standard Boolean query model. 
Consider a promise set $D\subseteq\{0,1\}^M$ and a hidden input string $x\in D$. 
Access to the individual bits of $x$ is provided by a standard Boolean oracle:
\begin{equation}\label{eq:boolean-oracle}
    O_x\ket{q,z}
    =
    \ket{q,z\oplus x_q},
    \qquad
    q\in\mathbb Z_M,
\end{equation}
where $z \in \{0, 1\}$.
When the overall Boolean input consists of $k$ separate bit strings, denoted as $X=(x^{(1)},\ldots,x^{(k)})\in D^k$, we use a combined Boolean oracle defined as:
\begin{equation}\label{eq:combined-boolean-oracle}
    O_X\ket{t,q,z}
    =
    \ket{t,q,z\oplus x^{(t)}_q},
    \qquad
    t\in[k],\quad q\in\mathbb Z_M.
\end{equation}
The register $t$ specifies which of the $k$ strings is being queried, and the register $q$ specifies the index within that string.  
Each use of either Boolean oracle or its inverse counts as one Boolean query.
Again, these oracles are self-inverse.

We will lower bound Boolean query complexity with the negative-weight adversary method~\cite{Hoyer2007}. 
Let $M\ge1$ be the number of input bits, let $D\subseteq\{0,1\}^{M}$ be a promise set, and let $f:D\to\{0,1\}$ be a partial Boolean function, meaning that $f$ is defined only on the promise set $D$, rather than on all $M$-bit strings.  
For every query position $q\in\mathbb Z_M$, we define a matrix $\Delta_q$ whose rows and columns are indexed by the inputs $x,y\in D$:
\begin{equation}\label{eq:delta-q}
    \Delta_q(x,y)
    =
    \begin{cases}
        1, & x_q\neq y_q,\\
        0, & x_q=y_q.
    \end{cases}
    \qquad x,y\in D.
\end{equation}
The matrix $\Delta_q$ simply marks with a $1$ which pairs of inputs disagree at position $q$, with no restriction on whether they also disagree elsewhere.
For a nonconstant function $f$, its negative-weight adversary bound, denoted by $\Adv^{\pm}(f)$, is defined formally as
\begin{equation}\label{eq:adv-def}
    \Adv^{\pm}(f)
    \coloneqq
    \max_{\Gamma}
    \frac{\norm{\Gamma}}
         {\displaystyle\max_{q\in\mathbb Z_M}
          \norm{\Gamma\circ\Delta_q}},
\end{equation}
where $\Gamma$ ranges over nonzero Hermitian matrices whose rows and columns are indexed by inputs $x,y\in D$ subject to the condition $\Gamma_{x,y}=0$ whenever $f(x)=f(y)$.
Intuitively, the numerator measures the size of the adversary matrix $\Gamma$, while the denominator measures the largest part that any single query can distinguish.
For our Boolean problem, we will calculate both norms explicitly.

The second query-complexity ingredient is the XOR lemma of Lee and Roland~\cite{Lee2012}.  
For $k$ inputs $x^{(1)},\ldots,x^{(k)}\in D$, we define the XOR of the function outputs as
\begin{equation}\label{eq:xor-function}
 \left(\bigoplus\circ f^{(k)}\right)(x^{(1)},\ldots,x^{(k)}) = f(x^{(1)})\oplus\cdots\oplus f(x^{(k)})
\end{equation}
Lee and Roland originally stated their XOR lemma using an additive adversary quantity $\Adv^\star(F_f)$ associated with the matrix $F_f(x,y)=1$ if $f(x)=f(y)$ and $F_f(x,y)=0$ otherwise. 
Their Remark~3.4 shows that the usual negative-weight adversary bound in~\eqref{eq:adv-def} satisfies the inequality $\Adv^\pm(f)\le \Adv^\star(F_f)$. 
Therefore, their Corollary~5.3 in Ref.~\cite{Lee2012} immediately gives the following form, which is the only version we use.

\begin{lemma}[XOR lemma~\cite{Lee2012}]\label{lem:xor}
Let $f:D\rightarrow\{0,1\}$ be a partial Boolean function, let $k\ge1$, and let $0\le\delta\le1$.  
Then
\begin{equation}\label{eq:xor}
    Q_{(1-\delta^{k/2})/2}
    \!\left(
        \bigoplus\circ f^{(k)}
    \right)
    \ge
    \frac{k\delta}{8}\,
    \Adv^{\pm}(f).
\end{equation}
\end{lemma}

\section{Main Result and Proof Sketch}\label{sec:strategy}

We are now ready to state the main result of this paper.

\begin{theorem}[Main Result]\label{thm:main} For $\kappa\ge31$, $s\ge5$, and $0<\varepsilon\le1/64$, the quantum query complexity of QLS (Problem~\ref{prob:qls}) in the standard sparse-access model is
$\Omega\!\left(\kappa\sqrt{s}\log(1/\varepsilon)\right)$. 
\end{theorem}

Before giving the full proof, we explain the construction and indicate exactly where the three factors $\sqrt{s}$, $\kappa$, and $\log(1/\varepsilon)$ arise.
The detailed argument is then developed in Sections~\ref{sec:position-parity}--\ref{sec:lower}.

The proof begins with a Boolean problem that we refer to as the \textsc{Position-Parity} problem.
An input to this problem is a bit string $x\in\{0,1\}^{m}$ with the promise that it contains exactly one $1$.
We denote the set of all such $m$-bit strings by $D_m$, which is the promise set for the problem.
The goal is to determine whether the position of this $1$ is an even or odd number.
Equivalently, if $j\in\mathbb{Z}_m$ is the unique index such that $x_j=1$, we define $g(x) \coloneqq j\bmod 2.$
Thus, given quantum query access to $x\in D_m$, the \textsc{Position-Parity} problem is to compute $g(x)$.
Using the negative-weight adversary method, we prove that this problem requires $\Omega(\sqrt{m})$ queries to $O_x$.

We next associate a unitary matrix with a single instance of the \textsc{Position-Parity} problem.
For an input string $x\in D_m$, we construct an $m$-dimensional unitary matrix $K_x$. 
We define two orthogonal states $\ket{0_L}$ and $\ket{1_L}$ that span a logical two-dimensional subspace, and construct the unitary $K_x$ such that its action on this subspace directly encodes the value $g(x)$:
\begin{equation}\label{eq:sketch-logical}
K_x\ket{z_L} = -\ket{(z\oplus g(x))_L}, \qquad z\in\{0,1\}.
\end{equation}
Thus, apart from an overall minus sign, one application of $K_x$ XORs the value $g(x)$ into the logical state.
Crucially, we construct the matrix $K_x$ such that the locations of its nonzero entries are entirely independent of the hidden bit string $x$. 
Only the numerical values of the entries depend on $x$, and each value depends on exactly one bit of $x$. 
This specific property allows us to simulate sparse-matrix queries to the final QLS instance using only a constant number of Boolean queries.

We next combine $n$ independent instances of the \textsc{Position-Parity} problem, denoted by $X=(x^{(1)},\ldots,x^{(n)})$, and define
\begin{equation}\label{eq:F-sketch}
    F(X)
    \coloneqq
    \bigoplus_{t=1}^{n} g\!\left(x^{(t)}\right).
\end{equation}
One can compute $F(X)$ by applying the unitaries $K_{x^{(1)}},\ldots,K_{x^{(n)}}$ in sequence. 
By~\eqref{eq:sketch-logical}, it is clear that each unitary $K_{x^{(t)}}$ XORs the corresponding value $g(x^{(t)})$ into the logical state, so starting from the state $\ket{0_L}$ the sequence produces the state $\ket{F(X)_L}$, up to the accumulated phase $(-1)^n$.
We choose $n$ to be an even number, so that $(-1)^n=1$ and the resulting logical state is exactly $\ket{F(X)_L}$.

To record the intermediate states of this computation, we construct a clock unitary $B_X$ that advances the computation through $3n$ steps. 
During the first $n$ steps, it applies $K_{x^{(1)}},\ldots,K_{x^{(n)}}$ sequentially as mentioned above, producing the state $\ket{F(X)_L}$. 
During the next $n$ steps, it applies the identity, so the logical register remains in the state $\ket{F(X)_L}$. 
During the final $n$ steps, it applies the inverse unitaries $K_{x^{(n)}}^\dagger,\ldots,K_{x^{(1)}}^\dagger$ in reverse order, thereby uncomputing the first stage and returning the logical register to $\ket{0_L}$.
Thus, repeated applications of $B_X$ generate the successive states visited during this $3n$-step computation, which we refer to as its computation history. After all $3n$ steps, the computation returns to its starting point, and we show that $B_X^{3n}=I$.

We next encode this computation history into a linear system.
For a decay parameter $0<r<1$, we define $A_X
\coloneqq
I-rB_X.$
Since $B_X$ is unitary, we have $\norm{rB_X}=r<1$, so $A_X$ is invertible.
Starting from the state $\ket{b_0}=\ket{0_L}\ket{0}_c$, where $\ket{0}_c$ denotes the initial clock state, each power $B_X^\ell\ket{b_0}$ gives the state reached after $\ell$ steps of the computation. 
Expanding $A_X^{-1}$ as a Neumann series and using the identity $B_X^{3n}=I$, we show that $A_X^{-1}\ket{b_0}$ is a weighted superposition of these $3n$ successive states, with the amplitude of the state at step $\ell$ proportional to $r^\ell$.
In this way, the solution of the linear system contains the entire computation history, including the middle $n$ steps during which the logical state is $\ket{F(X)_L}$. 

The same parameter $r$ also controls the condition number of $A_X$.
Choosing $r=1-1/R$ gives $\kappa(A_X)=\Theta(R)$.
Since $A_X$ need not be symmetric, we embed it into a real symmetric matrix $H_X$ while preserving the condition number.
We also show that the normalized QLS solution for $H_X$ contains the same weighted computation history as $A_X^{-1}\ket{b_0}$.
We later choose $R=\Theta(\kappa)$.

We now use the computation history contained in the QLS solution to recover $F(X)$.
By choosing $n$ appropriately, we ensure that the total probability of observing one of the middle clock values $\{n,\ldots,2n-1\}$ is greater than $6\varepsilon$. 
Throughout this entire interval, the logical register is exactly in the state $\ket{F(X)_L}$. 
Therefore, if measuring the clock gives a value in the middle interval, we can determine $F(X)$ with certainty by measuring the logical register; otherwise, we simply output a uniformly random bit.
For the exact QLS solution, this decoder succeeds with probability greater than $1/2+3\varepsilon$. 

If the QLS algorithm instead prepares an approximate state $\ket{\widetilde{x}}$ satisfying $\norm{\ket{\widetilde{x}}-\ket{x}}\le\varepsilon$, then the trace distance from the exact solution is at most $\varepsilon$, so the success probability of the decoder decreases by at most $\varepsilon$. Hence an $\varepsilon$-accurate QLS solution still allows us to compute $F(X)$ with success probability greater than $1/2+2\varepsilon$.

It remains to translate this decoding procedure into a query lower bound.
Applying the XOR lemma (Lemma~\ref{lem:xor}) to the $n$ independent \textsc{Position-Parity} instances shows that computing $F(X)$ with success probability greater than $1/2+2\varepsilon$ requires $\Omega(n\sqrt{m})$ Boolean queries.
On the other hand, by construction, sparse-matrix queries to $H_X$ can be simulated using only a constant number of Boolean queries to $X$. 
Therefore, if a QLS algorithm uses $T$ sparse-matrix queries, we must have $T=\Omega(n\sqrt{m})$.
To ensure that the probability of the middle clock interval remains larger than $6\varepsilon$, we choose the number of instances so that the decay $r^n$ is of order $\sqrt{\varepsilon}$. Since $r=1-1/\Theta(\kappa)$, this gives $n=\Theta(\kappa\log(1/\varepsilon))$.
Finally, choosing $m=\Theta(s)$ and substituting these relations into $T=\Omega(n\sqrt{m})$ gives $T = \Omega(\kappa\sqrt{s}\log(1/\varepsilon))$, which proves the claimed lower bound in~\eqref{eq:expected-lower-bound}.

\section{Query Complexity of \texorpdfstring{\textnormal{\textsc{Position-Parity}}}{Position-Parity}}
\label{sec:position-parity}

We begin the detailed proof by formally defining the \textsc{Position-Parity} problem and establishing its quantum query complexity.
Let $m\ge4$ be an even integer, and for each $j\in\mathbb{Z}_m$, let $e_j\in\{0,1\}^m$ denote the bit string whose unique $1$ occurs at position $j$.
We define the promise set
\begin{equation}
    D_m \coloneqq \{e_j:j\in\mathbb{Z}_m\}
\end{equation}
and the partial Boolean function
\begin{equation}\label{eq:g-def}
    g(e_j) \coloneqq j\bmod2.
\end{equation}
Thus, given quantum query access to an input $x\in D_m$ through the oracle $O_x$ defined in~\eqref{eq:boolean-oracle}, the \textsc{Position-Parity} problem is to determine $g(x)$, that is, to determine whether the position of the unique $1$ in $x$ is even or odd.
We now show that computing $g$ requires $\Omega(\sqrt{m})$ quantum queries.

We begin by constructing an explicit adversary matrix $\Gamma$ to lower bound $\Adv^{\pm}(g)$ defined previously in~\eqref{eq:adv-def}.
We index the rows and columns of $\Gamma$ by the positions
$j,k\in\mathbb{Z}_m$ and define
\begin{equation}\label{eq:Gamma-position-parity}
    \Gamma_{j,k}
    \coloneqq
    \begin{cases}
        1,
        &
        g(e_j)\neq g(e_k),
        \\[2mm]
        0,
        &
        g(e_j)=g(e_k).
    \end{cases}
\end{equation}
By construction, $\Gamma$ is a nonzero Hermitian matrix and $\Gamma_{j,k}=0$ whenever $g(e_j)=g(e_k)$, so $\Gamma$ is a valid adversary matrix for $g$ and is admissible in~\eqref{eq:adv-def}.

Since $m$ is even, the promise set $D_m$ contains exactly $m/2$ inputs $e_j$ with $g(e_j)=0$, corresponding to even positions $j$, and exactly $m/2$ inputs $e_j$ with $g(e_j)=1$, corresponding to odd positions $j$.
If the inputs are ordered so that those satisfying $g(e_j)=0$ come first and those satisfying $g(e_j)=1$ come second, then the rows and columns of $\Gamma$ can be arranged in the block form
\begin{equation}\label{eq:Gamma-block}
    \Gamma
    =
    \begin{pmatrix}
        0 & J_{m/2}\\
        J_{m/2} & 0
    \end{pmatrix},
\end{equation}
where $J_{m/2}$ denotes the $(m/2)\times(m/2)$ all-ones matrix. 
The square of this matrix is block diagonal:
\begin{equation}
    \Gamma^2 = \begin{pmatrix} J_{m/2}^2 & 0 \\ 0 & J_{m/2}^2 \end{pmatrix}.
\end{equation}
Each entry of $J_{m/2}^2$ is a sum of $m/2$ ones, so $J_{m/2}^2=(m/2)J_{m/2}$.
Since the all-ones matrix $J_d$ has a single nonzero eigenvalue equal to $d$, the matrix $J_{m/2}^2$ has a single nonzero eigenvalue equal to $(m/2)^2$.
It follows from the block form of $\Gamma^2$ that its two nonzero eigenvalues are both $(m/2)^2$.
Therefore, the two nonzero eigenvalues of $\Gamma$ are $\pm m/2$.
Hence the spectral norm of $\Gamma$ is
\begin{equation}\label{eq:Gamma-norm}
    \norm{\Gamma} = \frac{m}{2}.
\end{equation}

We next evaluate the norm of the Hadamard product $\norm{\Gamma\circ\Delta_q}$ for an arbitrary queried position $q\in\mathbb{Z}_m$.
Consider two input bit strings $e_j,e_k\in D_m$.
By definition, the $q$-th bit of $e_j$ satisfies $(e_j)_q=1$ if $j=q$ and $(e_j)_q=0$ otherwise, and the same holds for $e_k$.
Therefore, the two inputs differ at position $q$ if and only if exactly one of positions $j$ and $k$ is equal to $q$.
Equivalently, $\Delta_q(e_j,e_k)=1$ (see~\eqref{eq:delta-q} for the definition of $\Delta_q$) if and only if one of $j$ or $k$ equals $q$ and the other does not.

The matrix $\Gamma$ imposes one additional condition on these pairs: by~\eqref{eq:Gamma-position-parity}, it is clear that $\Gamma_{j,k}$ is nonzero only when $g(e_j)\neq g(e_k)$, or equivalently, when the positions $j$ and $k$ have opposite parity.
Therefore, $(\Gamma\circ\Delta_q)_{j,k}$ is nonzero exactly when one of $j$ or $k$ is equal to $q$ and the other has parity opposite to that of $q$.

Since $m$ is even, exactly half of the positions in $\mathbb{Z}_m$ have parity opposite to $q$.
Hence there are exactly $m/2$ such positions, and therefore exactly $m/2$ corresponding inputs $e_j\in D_m$.
Consequently, the matrix $\Gamma\circ\Delta_q$ acts as the adjacency matrix of a star graph whose center corresponds to the input $e_q$ and whose $m/2$ leaves correspond to the inputs $e_j$ whose marked positions have parity opposite to $q$.

Since the adjacency matrix of a star graph with $d$ leaves has nonzero eigenvalues $\pm\sqrt{d}$, its operator norm is $\sqrt{d}$.
In our case, $\Gamma\circ\Delta_q$ is the adjacency matrix of a star graph with $d=m/2$ leaves.
Therefore, for every $q\in\mathbb{Z}_m$, we have
\begin{equation}\label{eq:Gamma-Delta-norm}
\norm{\Gamma\circ\Delta_q}
=
\sqrt{\frac{m}{2}}.
\end{equation}

Substituting~\eqref{eq:Gamma-norm} and~\eqref{eq:Gamma-Delta-norm}
into the definition of the adversary bound $\Adv^{\pm}(g)$ in~\eqref{eq:adv-def} gives the following lemma.

\begin{lemma}[\textsc{Position-Parity} adversary bound]\label{lem:colored}
For the function $g:D_m\rightarrow\{0,1\}$ defined in~\eqref{eq:g-def}, we have $\Adv^\pm(g) \ge \sqrt{m/2}$.
\end{lemma}

\begin{proof}
The matrix $\Gamma$ in~\eqref{eq:Gamma-position-parity} is feasible, so~\eqref{eq:adv-def},~\eqref{eq:Gamma-norm}, and~\eqref{eq:Gamma-Delta-norm} give
\begin{equation}
    \Adv^\pm(g) \ge \frac{\norm{\Gamma}}{\displaystyle \max_{q\in\mathbb{Z}_m} \norm{\Gamma\circ\Delta_q}} = \frac{m/2}{\sqrt{m/2}} = \sqrt{\frac{m}{2}}.
\end{equation}
\end{proof}

For the remainder of the proof, we will use the adversary bound in Lemma~\ref{lem:colored} directly, since this is the quantity required by the XOR lemma. 
Nevertheless, it also immediately implies the bounded-error quantum query complexity of the \textsc{Position-Parity} problem. 
For any fixed $0\le\eta<1/2$, set $k=1$ and $\delta=(1-2\eta)^2$ in Lemma~\ref{lem:xor}.
The error probability threshold is then $\eta$, and the lemma gives $Q_\eta(g)\ge(1-2\eta)^2\Adv^\pm(g)/8$.
The coefficient is a positive constant, so $Q_\eta(g)=\Omega(\sqrt m)$.
This settles the quantum query complexity of the $\textsc{Position-Parity}$ problem.

\section{A Unitary for a Single \texorpdfstring{\textnormal{\textsc{Position-Parity}}}{Position-Parity} Instance}\label{sec:selector}

We now construct a unitary $K_x$ associated with one input $x\in D_m$.
This matrix acts on an $m$-dimensional register with computational basis states $\{\ket a:a\in\mathbb Z_m\}$ and must satisfy two properties that are needed later in the reduction: its action on a suitable two-dimensional logical subspace encodes the value $g(x)$, while its nonzero locations remain independent of $x$ and each matrix entry depends on at most one bit of the input string $x$.

Let $S$ denote the cyclic shift on this $m$-dimensional register, defined as $S\ket{a}=\ket{a+1\bmod m}$ for all $a\in\mathbb{Z}_m$, and let the state
$\ket{u}=m^{-1/2}\sum_{a=0}^{m-1}\ket{a}$ denote the uniform superposition state over the computational basis states.
For every pair $a,b\in\mathbb{Z}_m$, observe that there is exactly one cyclic shift $S^j$ with power $j\in\mathbb{Z}_m$ such that $S^j\ket{b}=\ket{a}$.
Therefore, the sum of all powers of $S$ is the all-ones matrix.
This relation can be written using the state $\ket{u}$ as
\begin{equation}\label{eq:sum-shifts}
\sum_{j=0}^{m-1}S^j=m\ket{u}\!\bra{u}.
\end{equation}

With that in place, for an input $x\in D_m$, we define the matrix $K_x$ as the following linear combination of the powers of $S$:
\begin{equation}\label{eq:K-def}
    K_x
    \coloneqq
    \sum_{j=0}^{m-1}
    \left(
        x_j-\frac{2}{m}
    \right)
    S^{j+1}.
\end{equation}
Suppose the unique marked position is $j_*$, so $x=e_{j_*}$.
The part of~\eqref{eq:K-def} containing $x_j$ is then $S^{j_*+1}$.
Also, $S^m=I$, so the shifts $S^1,\ldots,S^m$ have the same sum as $S^0,\ldots,S^{m-1}$.
It follows that
\begin{equation}
    K_x
    =
    S^{j_*+1}
    -
    \frac{2}{m}
    \sum_{j=0}^{m-1}S^{j+1}
    =
    S^{j_*+1}
    -
    2\ket{u}\!\bra{u}.
    \label{eq:K-first-factor}
\end{equation}

Next, observe that the uniform state $\ket{u}$ is invariant under every cyclic shift.  
Indeed, applying $S^j$ merely permutes the computational basis states appearing with equal amplitude in $\ket{u}$, so $S^j\ket{u}=\ket{u}$ for every $j\in\mathbb{Z}_m$. 
In particular, $S^{j_*+1}\ket{u} = \ket{u}$, or equivalently $ S^{j_*+1}\ket{u}\!\bra{u} = \ket{u}\!\bra{u}$.
By substituting this identity into the expression above, we can factor $K_x$:
\begin{equation}\label{eq:K-factor}
    K_x
    =
    S^{j_*+1}
    \left(
        I-2\ket{u}\!\bra{u}
    \right).
\end{equation}
Written in this form, it is immediately clear that $K_x$ is unitary. 
It is the product of a reflection matrix $(I-2\ket{u}\!\bra{u})$ and a permutation matrix $S^{j_*+1}$, both of which are unitary.

We now examine the individual matrix entries of $K_x$. 
Fix $a,b\in\mathbb{Z}_m$.  
By the definition of the cyclic shift, $S^{j+1}\ket{b}=\ket{b+j+1\bmod m}$.  
Therefore $\bra{a}S^{j+1}\ket{b}$ is equal to $1$ precisely when $a\equiv b+j+1\pmod m$, or equivalently when $j\equiv a-b-1\pmod m$, and it is $0$ otherwise.  
For fixed $a$ and $b$, exactly one value of $j\in\mathbb{Z}_m$ satisfies this condition.  
Substituting this value into~\eqref{eq:K-def} gives
\begin{equation}\label{eq:K-entry}
    \bra{a}K_x\ket{b}
    =
    x_{a-b-1\bmod m}
    -
    \frac{2}{m}.
\end{equation}

The above equation makes the structure of $K_x$ clear.  
It tells us that every matrix entry is either $1-2/m$ (if the queried bit is $1$) or $-2/m$ (if the queried bit is $0$). Since $m\ge 4$, neither of these values is ever zero.  
Thus every row and every column of $K_x$ contains exactly $m$ nonzero entries.

More importantly, the locations of these nonzero entries do not depend on $x$.  
Changing the marked position changes only which one of the $m$ entries in a row has value $1-2/m$; it does not change the set of positions at which the entries are nonzero.  
This distinction is essential for the sparse-access reduction.
The sparse-position oracle can reveal only the locations of the nonzero entries, so in this construction it reveals no information about the marked position $j_*$.  
The hidden information appears only in the matrix values.  
Moreover,~\eqref{eq:K-entry} shows that the value of any entry $\bra{a}K_x\ket{b}$ is determined entirely by the single Boolean bit $x_{a-b-1\bmod m}$.

We summarize the properties of $K_x$ established above for later use.

\begin{lemma}[Properties of $K_x$]\label{prop:selector}
For every $x\in D_m$, the matrix $K_x$ in~\eqref{eq:K-def} is unitary and satisfies~\eqref{eq:K-factor}.
Moreover, its entries satisfy~\eqref{eq:K-entry}.
Hence every row and every column of $K_x$ contains exactly $m$ nonzero entries, and their locations are independent of $x$.
\end{lemma}

We will later use both $K_x$ and its inverse when we reverse the computation.  
Since all entries of $K_x$ are real, we have $K_x^\dagger=K_x^T$.  
Therefore, by exchanging $a$ and $b$ in~\eqref{eq:K-entry}, we get $\bra{a}K_x^\dagger\ket{b}=x_{b-a-1\bmod m}-2/m$.
Hence an entry of either $K_x$ or $K_x^\dagger$ is determined by a single bit of $x$.

It remains to explain how the action of $K_x$ encodes the answer $g(x)$.
To define this logical subspace, we separate the computational basis states according to the parity of their indices and set
\begin{equation}\label{eq:logical-states}
    \ket{0_L}
    \coloneqq
    \sqrt{\frac{2}{m}}
    \sum_{\substack{a\in\mathbb{Z}_m\\a\ {\rm even}}}
    \ket{a},
    \qquad
    \ket{1_L}
    \coloneqq
    \sqrt{\frac{2}{m}}
    \sum_{\substack{a\in\mathbb{Z}_m\\a\ {\rm odd}}}
    \ket{a}.
\end{equation}
Because $m$ is even, there are exactly $m/2$ even indices and $m/2$ odd indices.
The prefactor $\sqrt{2/m}$ therefore normalizes both states.
Moreover, $\ket{0_L}$ and $\ket{1_L}$ have support on disjoint sets of computational basis states, so they are orthogonal.
Hence $\{\ket{0_L},\ket{1_L}\}$ forms an orthonormal basis for a two-dimensional logical subspace.

We now determine how the two factors of $K_x$ (see~\eqref{eq:K-factor}) act on this subspace.
First consider the reflection $I-2\ket{u}\!\bra{u}$.
Since the even and odd basis states together contain all $m$ computational basis states, the uniform state can be written as
\begin{equation}\label{eq:u-logical}
    \ket{u}
    =
    \frac{\ket{0_L}+\ket{1_L}}{\sqrt{2}}.
\end{equation}
It follows that $\braket{u}{0_L}=\braket{u}{1_L}=1/\sqrt{2}$.
Therefore,
\begin{equation}
    \left(I-2\ket{u}\!\bra{u}\right)\ket{0_L}
    =
    \ket{0_L}
    -
    \sqrt{2}\ket{u}
    =
    -\ket{1_L},
\end{equation}
and, by the same calculation,
$\left(I-2\ket{u}\!\bra{u}\right)\ket{1_L}=-\ket{0_L}$.
Thus the reflection swaps the logical states and applies a minus sign:
\begin{equation}\label{eq:reflection-logical}
    \left(I-2\ket{u}\!\bra{u}\right)\ket{z_L}
    =
    -\ket{(z\oplus1)_L},
    \qquad
    z\in\{0,1\}.
\end{equation}

Next consider the cyclic shift $S^d$.
If $d$ is even, adding $d$ modulo $m$ preserves the parity of every basis index, so $S^d$ maps the even positions among themselves and the odd positions among themselves.
Hence it leaves both logical states unchanged.
If $d$ is odd, adding $d$ reverses the parity of every basis index, so $S^d$ exchanges the even and odd positions and therefore exchanges $\ket{0_L}$ and $\ket{1_L}$.
Both cases can be written compactly as
\begin{equation}\label{eq:shift-logical}
    S^d\ket{z_L}
    =
    \ket{(z\oplus(d\bmod2))_L}.
\end{equation}

We can now combine these two actions to determine the action of $K_x$.
Let $x=e_{j_*}$.
From~\eqref{eq:K-factor}, $K_x$ first applies the reflection $I-2\ket{u}\!\bra{u}$ and then the shift $S^{j_*+1}$.
The reflection $I-2\ket{u}\!\bra{u}$ maps $\ket{z_L}$ to $-\ket{(z\oplus1)_L}$.
The subsequent shift changes the logical bit by the parity of $j_*+1$.
Since $g(x)=j_*\bmod2$, we have
$(j_*+1)\bmod2=g(x)\oplus1$.
Therefore the two logical flips combine as
\begin{equation}
    1\oplus\bigl((j_*+1)\bmod2\bigr)
    =
    1\oplus g(x)\oplus1
    =
    g(x).
\end{equation}
Substituting this into the action of the two factors gives
\begin{equation}\label{eq:logical-action}
    K_x\ket{z_L}
    =
    -\ket{(z\oplus g(x))_L},
    \qquad
    z\in\{0,1\}.
\end{equation}

Thus, when restricted to the logical subspace
$\operatorname{span}\{\ket{0_L},\ket{1_L}\}$, the unitary $K_x$ XORs the value $g(x)$ into the logical state.
The remaining minus sign is independent of $z$ and is therefore an overall phase for a single application of $K_x$.

\section{A Clock Unitary for \texorpdfstring{$n$}{n} \texorpdfstring{\textnormal{\textsc{Position-Parity}}}{Position-Parity} Instances}\label{sec:clock}

We next combine $n$ independent instances of the \textsc{Position-Parity} problem.
Let $X=(x^{(1)},\ldots,x^{(n)})\in D_m^n$ denote these $n$ \textsc{Position-Parity} instances.
The goal is to compute the XOR of all $n$ instances, that is,
\begin{equation}\label{eq:def-FX}
F(X)=\bigoplus_{t=1}^{n}g(x^{(t)}).
\end{equation}
We take $n$ to be even, with its specific choice given later in~\eqref{eq:n-choice}, and construct a $3n$-step reversible computation that first computes $F(X)$, then keeps the resulting logical state unchanged, and finally uncomputes the first stage.

The computation acts on two registers.
The first is the $m$-dimensional register introduced in the previous section, whose logical subspace is spanned by $\{\ket{0_L},\ket{1_L}\}$ and on which the unitaries $K_{x^{(t)}}$ act.
We refer to this register as the logical register in what follows.
We introduce a second register, called the clock register, with basis states
$\{\ket{t}_c:0\le t<3n\}$.
The clock value $t$ records the current step of the computation.
For each $t\in\{0,\ldots,3n-1\}$, let $L_t$ denote the operation applied to the first register as the clock advances from $t$ to $t+1$:
\begin{equation}\label{eq:L-def}
    L_t
    \coloneqq
    \begin{cases}
        K_{x^{(t+1)}}, & 0\le t<n,\\[1mm]
        I, & n\le t<2n,\\[1mm]
        K_{x^{(3n-t)}}^\dagger, & 2n\le t<3n.
    \end{cases}
\end{equation}
Thus, the first $n$ steps apply $K_{x^{(1)}},\ldots,K_{x^{(n)}}$ in sequence, the next $n$ steps apply the identity, and the final $n$ steps uncompute the first $n$ steps by applying $K_{x^{(n)}}^\dagger,\ldots,K_{x^{(1)}}^\dagger$ in reverse order.

We combine these $3n$ operations into a single clock unitary $B_X$, defined as follows, by coupling the logical register to the clock register:
\begin{equation}\label{eq:B-def}
    B_X
    \coloneqq
    \sum_{t=0}^{3n-1}
    L_t\otimes
    \ket{t+1\bmod 3n}\!\bra{t}_c.
\end{equation}
Thus, if the logical register is in an arbitrary state $\ket{\phi}$ and the clock register is in state $\ket{t}$, then $B_X\ket{\phi}\ket{t}_c=L_t\ket{\phi}\ket{t+1\bmod3n}_c$.
Hence one application of $B_X$ applies the operation $L_t$ associated with the current clock value $t$ and then advances the clock by one step.

Note that each $L_t$ is unitary: in the first and final intervals this follows from the unitarity of $K_x$, and in the middle interval this follows from the fact that $L_t=I$.
Moreover, the clock update $t\mapsto t+1\bmod3n$ is a permutation of the clock basis states.
These two facts already imply that $B_X$ is unitary.
Indeed, using the orthogonality of the clock basis states, the product $B_X^\dagger B_X$ reduces to
\begin{equation}
    \sum_{t=0}^{3n-1}L_t^\dagger L_t\otimes\ket{t}\!\bra{t}_c=I,
\end{equation}
and similarly $B_XB_X^\dagger=I$.

We next verify the periodicity property for $B_X$ which is $B_X^{3n}=I$.
We will use this property later in the next section when we invert the matrix $I-rB_X$.
Starting from clock value $0$, one complete cycle applies the operations
$L_0,L_1,\ldots,L_{3n-1}$ sequentially.
Because the final stage is the inverse of the first stage in reverse order, their ordered product is
\begin{equation}
    L_{3n-1}\cdots L_1L_0
    =
    K_{x^{(1)}}^\dagger\cdots K_{x^{(n)}}^\dagger
    K_{x^{(n)}}\cdots K_{x^{(1)}}
    =
    I.
\end{equation}
Thus one complete traversal of the clock returns both the clock and logical registers to their starting states when the initial clock value is $0$.

To conclude that $B_X^{3n}=I$, however, we must verify the same statement for every possible initial clock value.
Fix $t\in\mathbb{Z}_{3n}$ and write $P_t=L_{t-1}\cdots L_0$ and $Q_t=L_{3n-1}\cdots L_t$, with $P_0=I$.
The above identity then gives $Q_tP_t=I$.
Because $P_t$ and $Q_t$ are both unitaries, this implies that $P_tQ_t=I$.
But $P_tQ_t$ is exactly the operation obtained by starting at clock value $t$, advancing through $t,t+1,\ldots,3n-1$, and then continuing through $0,\ldots,t-1$.
After these $3n$ steps, the clock has also returned to $t$.
Therefore
\begin{equation}\label{eq:B-period}
    B_X^{3n}=I
\end{equation}
on the full logical-clock Hilbert space.

We now follow the particular computation history that will appear in the QLS solution.
Starting from the initial quantum state
\begin{equation}\label{eq:b0-def}
    \ket{b_0}
    \coloneqq
    \ket{0_L}\ket{0}_c,
\end{equation}
we define the sequence of states
\begin{equation}\label{eq:beta-def}
    \ket{\beta_\ell}
    \coloneqq
    B_X^\ell\ket{b_0},
    \qquad
    0\le\ell<3n.
\end{equation}
We refer to the states $\ket{\beta_0},\ldots,\ket{\beta_{3n-1}}$ as the computation-history states.
By construction, $\ket{\beta_\ell}$ has clock register in the state $\ket{\ell}_c$.
Hence the states $\ket{\beta_0},\ldots,\ket{\beta_{3n-1}}$ are mutually orthogonal.
Since $B_X$ is unitary and $\ket{b_0}$ is normalized, each $\ket{\beta_\ell}$ is also normalized, so these states form an orthonormal set.

It remains to determine the state of the logical register during the middle interval $t\in\{n,\ldots,2n-1\}$, where the computation will contain the target Boolean value $F(X)$ defined in~\eqref{eq:def-FX}.
After the first $n$ clock steps, the logical register is in state
$K_{x^{(n)}}\cdots K_{x^{(1)}}\ket{0_L}$.
By~\eqref{eq:logical-action}, each application of
$K_{x^{(t)}}$ XORs the bit $g(x^{(t)})$ into the logical state and contributes a factor of $-1$.
Applying the $n$ unitaries successively therefore gives
\begin{equation}
    K_{x^{(n)}}\cdots K_{x^{(1)}}\ket{0_L}
    =
    (-1)^n
    \ket{
        (
            \bigoplus_{t=1}^{n}g(x^{(t)})
        )_L
    }
    =
    (-1)^n\ket{F(X)_L}.
\end{equation}
Since $n$ is an even number, the phase is $(-1)^n=1$, and hence
$\ket{\beta_n}=\ket{F(X)_L}\ket{n}_c$.

For the next $n$ clock steps,~\eqref{eq:L-def} gives $L_t=I$.
The logical register therefore remains in the state $\ket{F(X)_L}$ throughout the entire middle interval:
\begin{equation}\label{eq:middle-beta}
    \ket{\beta_\ell}
    =
    \ket{F(X)_L}\ket{\ell}_c,
    \qquad
    n\le\ell<2n.
\end{equation}

\section{The Hard Sparse Linear System}\label{sec:linear}

We now convert the clock unitary $B_X$ constructed in Section~\ref{sec:clock} into a linear system. 
The key observation is that powers of $B_X$ generate the computation-history states $\ket{\beta_\ell}$ defined in~\eqref{eq:beta-def}.
We therefore construct a matrix $A_X$ such that its inverse $A_X^{-1}$ can be expanded as a series of powers of $B_X$.
When this inverse is applied to $\ket{b_0}$, each power $B_X^\ell$ produces the corresponding computation-history state $\ket{\beta_\ell}$, so the solution of the linear system $A_X\ket{y}=\ket{b_0}$ becomes a weighted superposition of the computation-history states.

For a decay parameter $0<r<1$, define
\begin{equation}\label{eq:A-def}
A_X
\coloneqq
I-rB_X.
\end{equation}
Since $B_X$ is unitary, $\norm{rB_X}=r<1$.
Therefore $A_X$ is invertible, and its inverse is given by the convergent Neumann series
\begin{equation}
    A_X^{-1}=\sum_{t=0}^{\infty}r^tB_X^t
\end{equation}

We apply this inverse $A_X^{-1}$ to the initial state $\ket{b_0}$ from~\eqref{eq:b0-def}.
Recall that $B_X^\ell\ket{b_0}=\ket{\beta_\ell}$ and that
$B_X^{3n}=I$ by~\eqref{eq:B-period}.
Every integer $t\ge0$ can be written uniquely as $t=3nk+\ell$, where
$k\ge0$ and $0\le\ell<3n$.
Consequently,
$B_X^t\ket{b_0}=B_X^\ell\ket{b_0}=\ket{\beta_\ell}$.
Grouping the infinite series according to these complete clock cycles gives
\begin{align}
    A_X^{-1}\ket{b_0}
    &=
    \sum_{k=0}^{\infty}
    \sum_{\ell=0}^{3n-1}
    r^{3nk+\ell}\ket{\beta_\ell}
    \nonumber\\
    &=
    \frac{1}{1-r^{3n}}
    \sum_{\ell=0}^{3n-1}
    r^\ell\ket{\beta_\ell}.
    \label{eq:history-state}
\end{align}
Thus the solution is a weighted superposition of the $3n$ computation-history states, with the amplitude of $\ket{\beta_\ell}$ proportional to $r^\ell$.

We next determine how much of this solution lies in the middle part of the computation history.
By~\eqref{eq:middle-beta}, for every $n\le\ell<2n$ we have $\ket{\beta_\ell}=\ket{F(X)_L}\ket{\ell}_c$.
These are therefore precisely the clock values for which the Boolean value $F(X)$ can be read directly from the logical state.

Let $p$ denote the probability that the normalized state corresponding to~\eqref{eq:history-state} has a clock value in the interval $\{n,\ldots,2n-1\}$.
Because the computation-history states are orthonormal, their squared amplitudes are proportional to $r^{2\ell}$.
Also, the common normalization factor $1/(1-r^{3n})$ cancels, so
\begin{equation}
    p = \frac{\sum_{\ell=n}^{2n-1}r^{2\ell}}{\sum_{\ell=0}^{3n-1}r^{2\ell}}.
\end{equation}
Evaluating the two geometric sums gives
$p=r^{2n}(1-r^{2n})/(1-r^{6n})$.
Using
$1-r^{6n}=(1-r^{2n})(1+r^{2n}+r^{4n})$, we obtain
\begin{equation}\label{eq:p-formula}
    p
    =
    \frac{r^{2n}}
         {1+r^{2n}+r^{4n}}.
\end{equation}

We now choose $r$ and $n$ so that the probability in~\eqref{eq:p-formula} remains larger than the error allowed in the QLS solution.
At the same time, the choice of $r$ will determine the condition number of $A_X$.

Let $R\ge16$ be a power of two and set
\begin{equation}\label{eq:r-def}
r \coloneqq 1-\frac{1}{R}.
\end{equation}
Thus $0<r<1$ and $r$ approaches $1$ as $R$ increases.
It is convenient to define $\mu\coloneqq-\log r$.
Applying the elementary bound $x\le-\log(1-x)\le x/(1-x)$ with $x=1/R$ gives
\begin{equation}\label{eq:mu-bounds} 
\frac{1}{R} \le \mu \le \frac{1}{R-1}.
\end{equation} 
Hence $\mu=\Theta(1/R)$, while $r^n=e^{-\mu n}$.

For $0<\varepsilon\le1/64$, we choose the number of \textsc{Position-Parity} instances to be
\begin{equation}\label{eq:n-choice}
n \coloneqq 2 \left\lfloor \frac{\log(1/(9\varepsilon))}{4\mu} \right\rfloor.
\end{equation}
In particular, note that $n$ is even, as required in Section~\ref{sec:clock}.
We next get the upper and lower bounds on $n$, since this is where the factor $\log(1/\varepsilon)$ first enters the construction.

For the upper bound on $n$,~\eqref{eq:n-choice} gives $n\le\log(1/(9\varepsilon))/(2\mu)<\log(1/\varepsilon)/(2\mu)$.
Using $\mu\ge1/R$ from~\eqref{eq:mu-bounds}, we obtain $n<(R/2)\log(1/\varepsilon)$.

For the lower bound on $n$, the inequality $\lfloor y\rfloor>y-1$ gives $n>\log(1/(9\varepsilon))/(2\mu)-2$.
Since $\varepsilon\le1/64$, we have $\log(1/\varepsilon)\ge\log64$ and $\log9<(3/5)\log64\le(3/5)\log(1/\varepsilon)$.
Therefore $\log(1/(9\varepsilon))>(2/5)\log(1/\varepsilon)$.
Together with $1/\mu\ge R-1$ from~\eqref{eq:mu-bounds}, this gives $n>((R-1)/5)\log(1/\varepsilon)-2$.
Because $R\ge16$, we have $R-1\ge15R/16$, and hence $n>(3R/16)\log(1/\varepsilon)-2$.
Finally, $R\log(1/\varepsilon)\ge16\log64>32$, so $2<(R/16)\log(1/\varepsilon)$.
Combining both these bounds yields
\begin{equation}\label{eq:n-main-bounds}
    \frac{R}{8}\log\frac{1}{\varepsilon} < n < \frac{R}{2}\log\frac{1}{\varepsilon}.
\end{equation}
In particular, $n=\Theta(R\log(1/\varepsilon))$.
Later, after choosing $R=\Theta(\kappa)$, this will become the desired $n=\Theta(\kappa\log(1/\varepsilon))$.

We now verify that this choice of $n$ gives the required probability in the middle clock interval.
From~\eqref{eq:n-choice}, $n\le\log(1/(9\varepsilon))/(2\mu)$.
Since $r=e^{-\mu}$, it follows that $r^n=e^{-\mu n}\ge3\sqrt{\varepsilon}$.
We also need an upper bound on $r^n$ in order to control the denominator of~\eqref{eq:p-formula}.
Using the lower estimate from the floor function in~\eqref{eq:n-choice}, we have $n>\log(1/(9\varepsilon))/(2\mu)-2$, and therefore $r^n<3e^{2\mu}\sqrt{\varepsilon}$.
Equation~\eqref{eq:mu-bounds} and $R\ge16$ then imply $\mu\le1/15$.
Together with $\varepsilon\le1/64$, this gives $r^n<(3/8)e^{2/15}<1/2$.

Hence $r^{2n}<1/4$ and $r^{4n}<1/16$, so the denominator in~\eqref{eq:p-formula} is smaller than $1+1/4+1/16<3/2$.
Using also $r^n\ge3\sqrt{\varepsilon}$, we conclude that
\begin{equation}\label{eq:p-lower}
    p > \frac{2}{3}r^{2n} \ge 6\varepsilon.
\end{equation}
Thus, although the amplitudes in~\eqref{eq:history-state} decay as the clock advances, the total probability of the middle interval carrying $\ket{F(X)_L}$ remains strictly larger than $6\varepsilon$.

We next determine the condition number of $A_X$.
Let $\ket{v}$ be an eigenvector of $B_X$ with eigenvalue $\lambda$.
Since $B_X$ is unitary, $|\lambda|=1$, and~\eqref{eq:A-def} gives $A_X\ket{v}=(1-r\lambda)\ket{v}$.
Moreover, $A_X$ is a polynomial in the unitary $B_X$ and is therefore normal.
Its singular values are consequently the quantities $|1-r\lambda|$. 
For every $|\lambda|=1$, the triangle and reverse triangle inequalities give $1-r\le|1-r\lambda|\le1+r$.

It remains to check that both endpoints actually occur.
On the subspace spanned by the computation-history states
$\ket{\beta_0},\ldots,\ket{\beta_{3n-1}}$, the unitary $B_X$ acts as a cyclic shift of length $3n$.
The eigenvalues of such a shift are the $3n$-th roots of unity.
In particular, $+1$ is always an eigenvalue, and since $n$ is even, $3n$ is even, so $\e^{\iu \pi} = -1$ is also an eigenvalue.
Therefore the two extreme singular values of $A_X$ are attained, and
\begin{equation}\label{eq:kappa-A}
    \sigma_{\min}(A_X) = 1-r = \frac{1}{R},
    \qquad \sigma_{\max}(A_X) = 1+r,
    \qquad \kappa(A_X) = 2R-1.
\end{equation}

The matrix $A_X$ is real but need not be symmetric.
We therefore convert it into the real symmetric instance used in the final QLS reduction.
For that, we define $\overline A_X=A_X/2$ and
\begin{equation}\label{eq:H-def}
    H_X \coloneqq
    \begin{pmatrix}
        0 & \overline A_X\\
        \overline A_X^T & 0
    \end{pmatrix},
    \qquad \ket{b} \coloneqq
    \begin{pmatrix}
        \ket{b_0}\\
        0
    \end{pmatrix}.
\end{equation}
The factor $1/2$ ensures that $\norm{H_X}<1$, which is just a convenient normalization, and it does not change the condition number.

To see the spectrum of $H_X$, note that $H_X^2$ is block diagonal with blocks $\overline A_X\overline A_X^T$ and $\overline A_X^T\overline A_X$.
Its eigenvalues are therefore the squared singular values of $\overline A_X$, and the eigenvalues of $H_X$ are $\pm\sigma_j(\overline A_X)$.
Here $\sigma_j(A)$ denotes the $j$-th singular value of $A$.
Consequently,
\begin{equation}\label{eq:H-kappa}
    \kappa(H_X) = \kappa(\overline A_X) = \kappa(A_X) = 2R-1,
\end{equation}
while $\norm{H_X}=\norm{\overline A_X}=(1+r)/2<1$.

We also verify that such embedding of the matrix $A_X$ into $H_X$ preserves the solution state that contains the computation history.
Write $H_X^{-1}\ket b=(\ket{y_1},\ket{y_2})^T$.
Then $H_X(\ket{y_1},\ket{y_2})^T=\ket b$ is equivalent to $\overline A_X\ket{y_2}=\ket{b_0}$ and $\overline A_X^T\ket{y_1}=0$.
Since $\overline A_X$ is invertible, we must have $\ket{y_1}=0$ and $\ket{y_2}=\overline A_X^{-1}\ket{b_0}=2A_X^{-1}\ket{b_0}$.
Hence
\begin{equation}\label{eq:H-inverse-b}
    H_X^{-1}\ket b = 2
    \begin{pmatrix}
        0\\
        A_X^{-1}\ket{b_0}
    \end{pmatrix}.
\end{equation}
Note that the factor $2$ disappears upon normalization.
Therefore, the nonzero block of the normalized QLS solution for $H_X$ is exactly the normalized weighted superposition of computation-history states in~\eqref{eq:history-state}.

We finally verify that $H_X$ can be accessed efficiently in the sparse-access model.

\begin{lemma}\label{lem:sparse}
For every $X=(x^{(1)},\ldots,x^{(n)})\in D_m^n$, every row and every column of $H_X$ contains at most $m+1$ nonzero entries.
Moreover, the sparse-position oracle $O_{H_X}^{\mathrm{pos}}$ can be implemented without querying $O_X$, while $O_{H_X}^{\mathrm{val}}$ and its inverse can each be implemented using at most two queries to $O_X$.
If $m$ and $R$ are powers of two, every entry of $H_X$ is an exact dyadic rational.
\end{lemma}

\begin{proof}
It is enough to analyze the matrix $\overline A_X=(I-rB_X)/2$, since $H_X$ consists of $\overline A_X$ and its transpose $\overline A_X^T$ in the two off-diagonal blocks (see~\eqref{eq:H-def}).
Let us begin by fixing some notation.
For logical-register indices $a,b\in\mathbb{Z}_m$ and clock-register indices $t,t'\in\mathbb{Z}_{3n}$, we label the rows of $\overline A_X$ by pairs $(a,t')$ and its columns by pairs $(b,t)$.
With this notation in place,~\eqref{eq:A-def} and~\eqref{eq:B-def} give
\begin{equation}\label{eq:Abar-entry}
    (\overline A_X)_{(a,t'),(b,t)}
    =
    \frac{1}{2}\delta_{t',t}\delta_{a,b}
    -
    \frac{r}{2}
    \delta_{t',t+1\bmod3n}
    \bra{a}L_t\ket{b}.
\end{equation}
The first term comes from the identity matrix and is nonzero only when the row and column labels are the same, that is, when $t'=t$ and $a=b$.
The second term comes from $B_X$ and is nonzero only when the row clock-register index $t'$ is the successor of the column clock-register index $t$, namely $t'=t+1\bmod 3n$. 
For this pair of clock-register indices, the corresponding matrix element on the logical register is $\bra{a}L_t\ket{b}$, where $L_t$ is the operation defined in~\eqref{eq:L-def}.

Now recall that we already determined the matrix entries $\bra{a}L_t\ket{b}$ in Section~\ref{sec:selector}.
During the first stage, $L_t=K_{x^{(t+1)}}$, so~\eqref{eq:K-entry} gives $\bra{a}L_t\ket{b} =x^{(t+1)}_{a-b-1\bmod m}-2/m$.
During the middle stage, $L_t=I$, so $\bra{a}L_t\ket{b}=\delta_{a,b}$.
During the final stage, $L_t=K_{x^{(3n-t)}}^\dagger$, and the entry formula following Lemma~\ref{prop:selector} gives $\bra{a}L_t\ket{b}=x^{(3n-t)}_{b-a-1\bmod m}-2/m$.

We first bound the number of nonzero entries in each row of $\overline A_X$.
Fix a row $(a,t')$.
From~\eqref{eq:Abar-entry}, the first term, $\frac{1}{2}\delta_{t',t}\delta_{a,b}$, is nonzero only when $t=t'$ and $b=a$.
Thus, for the fixed row $(a,t')$, this term contributes exactly one possible nonzero entry, namely the diagonal entry in the column~$(a,t')$.

Now consider the second term in~\eqref{eq:Abar-entry}, that is, $-\frac{r}{2}\delta_{t',t+1\bmod3n}\bra{a}L_t\ket{b}$.
For the fixed row clock-register index $t'$, the factor $\delta_{t',t+1\bmod3n}$ can be nonzero for only one value of the column clock-register index, namely $t=t'-1\bmod3n$.
Therefore, for the fixed row $(a,t')$, the second term can be nonzero only in the $m$ columns $(b,t'-1\bmod3n)$ with $b\in\mathbb{Z}_m$.
For these columns, whether the corresponding matrix entry is nonzero is determined entirely by $\bra{a}L_{t'-1}\ket{b}$.

Suppose first that the clock-register index $t'-1\bmod3n$ lies in the first or final stage of the computation.
Then~\eqref{eq:L-def} shows that $L_{t'-1}$ is either some $K_x$ or some $K_x^\dagger$.
By Lemma~\ref{prop:selector}, every row of both $K_x$ and $K_x^\dagger$ contains exactly $m$ nonzero entries.
Hence, as $b$ ranges over $\mathbb{Z}_m$, the second term in~\eqref{eq:Abar-entry} contributes at most $m$ nonzero entries to the fixed row $(a,t')$.
Together with the single diagonal entry from the first term, this gives at most $m+1$ nonzero entries.

If instead the clock-register index $t'-1\bmod3n$ lies in the middle stage, then from~\eqref{eq:L-def}, we have that $L_{t'-1}=I$.
In this case, $\bra{a}L_{t'-1}\ket{b}=\delta_{a,b}$, so the second term is nonzero only for $b=a$.
The row then contains at most one nonzero entry from the second term together with the single diagonal entry from the first term, and therefore certainly no more than $m+1$ nonzero entries.
We conclude that every row of $\overline A_X$ contains at most $m+1$ nonzero entries.

The same argument applies to the columns of $\overline A_X$. 
For a fixed column $(b,t)$, the first term in~\eqref{eq:Abar-entry} contributes at most one nonzero entry, while the second term can be nonzero only for the unique clock-register index $t'=t+1\bmod 3n$. 
The remaining dependence is therefore determined by a single column of $L_t$. 
By Lemma~\ref{prop:selector}, every column of $K_x$ and $K_x^\dagger$ contains exactly $m$ nonzero entries, while every column of $I$ contains exactly one. 
Hence every column of $\overline A_X$ also contains at most $m+1$ nonzero entries.

It follows immediately that $\overline A_X^T$ has the same row and column sparsity. Since $H_X$ contains $\overline A_X$ and $\overline A_X^T$ in separate off-diagonal blocks, each row and each column of $H_X$ receives entries from only one of these blocks. Therefore every row and every column of $H_X$ contains at most $m+1$ nonzero entries.

We next show how the sparse-position oracle for $H_X$, that is, $O_{H_X}^{\mathrm{pos}}$, can be implemented without querying the Boolean oracle $O_X$.
Again since $H_X$ contains only $\overline A_X$ and $\overline A_X^T$ in its two off-diagonal blocks, it is enough to understand the locations of the nonzero entries of $\overline A_X$.

Consider a row $(a,t')$ of $\overline A_X$.
As shown above, the first term in~\eqref{eq:Abar-entry} contributes the diagonal entry in column $(a,t')$.
The second term can have support only on columns whose clock-register index is $t=t'-1\bmod 3n$.
The corresponding logical-register indices are determined by the support of the matrix $L_t$.

If $t$ lies in the first or final stage of the computation, then $L_t$ is some $K_x$ or $K_x^\dagger$.
By Lemma~\ref{prop:selector}, every row and every column of these matrices contains all $m$ logical-register indices as nonzero entries, independently of the input string $x$.
Thus the locations of these nonzero entries depend only on the row label $(a,t')$ and on the fixed definition of the clock construction, and not on any bit of $X$.

If $t$ lies in the middle stage, then $L_t=I$.
In this case the second term is nonzero only at the logical-register index $b=a$.
Again, this location is completely determined by the matrix indices and is independent of $X$.
If a row contains fewer than $m+1$ nonzero entries, the remaining sparse indices may be assigned to predetermined zero entries, as allowed by the sparse-access model.

Therefore, given a row label and a sparse index, the corresponding column label of a nonzero entry of $\overline A_X$ can be computed without querying $O_X$.
The same is true for $\overline A_X^T$, since its support is obtained by transposing the support of $\overline A_X$.
Since the locations of the two off-diagonal blocks of $H_X$ are also fixed by~\eqref{eq:H-def}, the sparse-position oracle $O_{H_X}^{\mathrm{pos}}$ can be implemented using no queries to $O_X$.

We now consider the sparse-value oracle $O_{H_X}^{\mathrm{val}}$.
Again, it is enough to analyze the entries of $\overline A_X$, because an entry of $\overline A_X^T$ is obtained by exchanging its row and column labels.
Suppose that the requested entry of $\overline A_X$ has row label $(a,t')$ and column label $(b,t)$.
Equation~\eqref{eq:Abar-entry} determines its value completely.
If $t'=t$ and $a=b$, the identity term contributes $1/2$, which is independent of $X$.
If $t'\neq t+1\bmod 3n$, then the second term vanishes.
Hence all entries that arise only from the identity term, as well as all zero entries, can be computed without querying $O_X$.

The only input-dependent case occurs when $t'=t+1\bmod 3n$ and $L_t$ belongs to the first or final stage.
If $0\le t<n$, then $L_t=K_{x^{(t+1)}}$, and~\eqref{eq:K-entry} gives
$\bra{a}L_t\ket{b}=x^{(t+1)}_{a-b-1\bmod m}-2/m$.
Thus the requested matrix entry depends on exactly one Boolean bit, namely $x^{(t+1)}_{a-b-1\bmod m}$.
If $2n\le t<3n$, then $L_t=K_{x^{(3n-t)}}^\dagger$, and the entry formula following Lemma~\ref{prop:selector} gives
$\bra{a}L_t\ket{b}=x^{(3n-t)}_{b-a-1\bmod m}-2/m$.
Thus in this case the requested matrix entry depends on exactly one Boolean bit, namely $x^{(3n-t)}_{b-a-1\bmod m}$.
Finally, if $n\le t<2n$, then $L_t=I$, so $\bra{a}L_t\ket{b}=\delta_{a,b}$ and the entry is again independent of $X$.
Therefore every entry of $\overline A_X$, and hence every entry of $H_X$, is either independent of the Boolean input or is determined by a single bit of one of the strings $x^{(1)},\ldots,x^{(n)}$.

We can therefore simulate one query to $O_{H_X}^{\mathrm{val}}$ as follows.
Given the row and column labels of the requested entry, we first determine whether the entry is independent of $X$.
If it is, its value can be computed directly without querying $O_X$.
Otherwise, the row and column labels determine a unique pair $(j,q)$ such that the entry depends on the Boolean bit $x^{(j)}_q$.
For the first clock interval, $j=t+1$ and
$q=a-b-1\bmod m$.
For the final clock interval, $j=3n-t$ and
$q=b-a-1\bmod m$.

We compute the pair $(j,q)$ into ancillary registers, choosing a fixed valid pair, say $(1,0)$, whenever the requested entry is independent of $X$.
We then apply $O_X$ once to write $x^{(j)}_q$ into an ancilla.
Using the row and column labels, the known parameters $r$ and $m$, and this bit when needed, we XOR the binary representation of the requested matrix entry into the value register.
For entries independent of $X$, this computation ignores the queried bit.
We apply $O_X$ a second time to return the bit ancilla to $\ket{0}$, and then uncompute the address registers without further queries.
This procedure acts coherently on superpositions of row and column labels and uses exactly two queries to $O_X$.

It remains to verify that the entries of $H_X$ admit the exact binary representation required by the sparse-value oracle.
When $m$ and $R$ are powers of two, the quantities
$2/m$, $1/R$, $r=1-1/R$, and $r/2$ are dyadic rationals.
Equation~\eqref{eq:Abar-entry}, together with the formulas for the entries of $K_x$ and $K_x^\dagger$, then shows that every entry of $\overline A_X$ is a dyadic rational.
The same is true for $\overline A_X^T$ and hence for $H_X$.
Therefore every entry of $H_X$ has an exact finite binary representation.

This concludes the proof.
\end{proof}

The logical register has dimension $m$, while the clock register has dimension $3n$.
Hence $A_X$ acts on a space of dimension $3mn$, and the symmetric embedding doubles this dimension:
\begin{equation}\label{eq:H-dimension}
    \dim(H_X)=6mn.
\end{equation}

At this point all properties of the hard QLS instance have been established.
Its sparsity is at most $m+1$, its condition number is $2R-1$, and its exact normalized solution contains the computation-history states from~\eqref{eq:history-state}.
Most importantly,~\eqref{eq:p-lower} shows that the middle computation-history states, which carry $\ket{F(X)_L}$ by~\eqref{eq:middle-beta}, have total probability greater than $6\varepsilon$.
The next section shows that this is enough to recover $F(X)$ even when the QLS solution is prepared only to error $\varepsilon$.

\section{Decoding from an Approximate Solution}\label{sec:decode}

We now show how to recover the Boolean value $F(X)$ from the output of a QLS algorithm that prepares only an $\varepsilon$-approximation to the exact solution state.
Let $\ket{x}$ denote the exact normalized solution state of the QLS instance $(H_X,\ket b)$ constructed in Section~\ref{sec:linear}.
So, we have $\ket{x} = H_X^{-1} \ket{b}/\norm{H_X^{-1}\ket{b}}$.
By~\eqref{eq:H-inverse-b}, observe that the state $\ket{x}$ has support entirely on the lower block of the symmetric embedding, and within that block, it is exactly the normalized version of the state $A_X^{-1}\ket{b_0}$ whose expression is given in~\eqref{eq:history-state}.
Moreover, by~\eqref{eq:middle-beta}, whenever the clock register lies in the middle interval $\{n,\ldots,2n-1\}$, the logical register is exactly in the state $\ket{F(X)_L}$.
We will use these two properties to define a measurement procedure that decodes $F(X)$.

Recall from~\eqref{eq:H-def} that $H_X$ acts on the direct sum of two copies of the logical-clock Hilbert space.
We first perform the two-outcome projective measurement
\begin{equation}\label{eq:projective-meas}
    \Pi_{\mathrm{up}}
    =
    \begin{pmatrix}
        I&0\\
        0&0
    \end{pmatrix},
    \qquad
    \Pi_{\mathrm{low}}
    =
    \begin{pmatrix}
        0&0\\
        0&I
    \end{pmatrix}.
\end{equation}
If the outcome is $\Pi_{\mathrm{up}}$, we output an independent uniformly random bit.
If the outcome is $\Pi_{\mathrm{low}}$, we measure the clock register in the basis $\{\ket{t}_c:0\le t<3n\}$ and denote the outcome by $t$.

Suppose first that $n\le t<2n$.
For the exact solution,~\eqref{eq:middle-beta} shows that the logical register is then exactly in the state $\ket{F(X)_L}$.
We therefore measure the logical register in an orthonormal basis containing the two states
$\ket{0_L}$ and $\ket{1_L}$.
If the outcome is $\ket{z_L}$ for some $z\in\{0,1\}$, we output the bit $z$.
If the logical-register measurement gives an outcome orthogonal to $\operatorname{span}\{\ket{0_L},\ket{1_L}\}$, we instead output an independent uniformly random bit.
For the exact solution, this latter case has probability zero whenever $n\le t<2n$.

Now, if $t\notin\{n,\ldots,2n-1\}$, we do not measure the logical register and instead output an independent uniformly random bit.
This completes the implementation of the decoder.
The entire decoding procedure consists only of measurements and classical postprocessing, and therefore requires no additional queries to the sparse-matrix oracles.

We first evaluate this decoder on the exact normalized solution $\ket{x}$.
By~\eqref{eq:H-inverse-b}, the state $\ket{x}$ lies entirely in the lower block of the symmetric embedding.
Thus the first measurement, that is, the two-outcome projective measurement defined in~\eqref{eq:projective-meas}, produces the outcome $\Pi_{\mathrm{low}}$ with probability one.
Conditioned on this outcome, the state of the logical and clock registers is the normalized computation-history state obtained from~\eqref{eq:history-state}.
By~\eqref{eq:p-formula}, measuring the clock register produces an outcome $t\in\{n,\ldots,2n-1\}$ with probability $p$.
For every such value of $t$,~\eqref{eq:middle-beta} shows that the state of the logical register is exactly $\ket{F(X)_L}$.
Hence, conditioned on the clock lying in the middle interval, the logical-register measurement returns the bit $F(X)$ with probability one.

With the remaining probability $1-p$, measuring the clock register gives outcome $t \notin \{n,\ldots,2n-1\}$.
In this case, by definition of the decoder, we output an independent uniformly random bit and are therefore correct with probability $1/2$.
The total success probability of the decoder on the exact solution is consequently
\begin{equation}
    P_{\mathrm{exact}} = p\cdot 1 + (1-p)\cdot\frac12 = \frac12+\frac p2.
\end{equation}
Using the lower bound $p>6\varepsilon$ from~\eqref{eq:p-lower}, we obtain
\begin{equation}\label{eq:exact-decoder-success}
P_{\mathrm{exact}} > \frac12+3\varepsilon.
\end{equation}

We next replace the exact solution $\ket{x}$ by the state produced by an $\varepsilon$-accurate QLS algorithm.
Suppose that the algorithm prepares a normalized state $\ket{\widetilde{x}}$ satisfying $\norm{\ket{\widetilde{x}}-\ket{x}}\le\varepsilon$.
We need to determine how much the success probability of the decoder can change when its input is changed from $\ket{x}$ to $\ket{\widetilde{x}}$.

Let $c=\braket{x}{\widetilde{x}}$.
The trace distance between the corresponding pure quantum states is 
\begin{equation}
\frac{1}{2}\norm{\ket{\widetilde{x}}\!\bra{\widetilde{x}}-\ket{x}\!\bra{x}}_1=\sqrt{1-|c|^2}.
\end{equation}
On the other hand, $\norm{\ket{\widetilde{x}}-\ket{x}}^2=2-2\operatorname{Re}(c)$.
Since $1-|c|^2=(1-|c|)(1+|c|)\le2(1-|c|)\le2(1-\operatorname{Re}(c))$, we obtain
\begin{equation}\label{eq:trace-euclidean}
    \frac12
    \norm{
        \ket{\widetilde{x}}\!\bra{\widetilde{x}}
        -
        \ket{x}\!\bra{x}
    }_1
    \le
    \norm{\ket{\widetilde{x}}-\ket{x}}
    \le
    \varepsilon.
\end{equation}

We know that the probability of any fixed measurement event can change by at most the trace distance between the two states.
In particular, the complete decoding procedure above can be regarded as a fixed measurement whose successful outcome is the event that the output equals $F(X)$.
It follows from~\eqref{eq:trace-euclidean} that the success probability of the decoder on $\ket{\widetilde{x}}$ is at least $P_{\mathrm{exact}}-\varepsilon$.
Combining this with~\eqref{eq:exact-decoder-success} gives 
\begin{equation}
    P_{\mathrm{approx}}>1/2+2\varepsilon.
\end{equation}

We have therefore proved the following statement.

\begin{lemma}[QLS decoder]\label{lem:decoder}
An $\varepsilon$-accurate QLS algorithm for the constructed instance $(H_X, \ket b)$ can be followed by a query-free decoding procedure that computes $F(X) = \bigoplus_{t=1}^{n}g(x^{(t)})$ with success probability strictly larger than $1/2+2\varepsilon$.
\end{lemma}

The important point is that the decoder does not recover $F(X)$ with constant advantage over random guessing.
Its advantage is only of order $\varepsilon$.
This is exactly the regime for which the XOR lemma in Lemma~\ref{lem:xor} is needed.
We now use that lemma to convert the decoding procedure into the final Boolean query lower bound.

\section{Proof of the Query Lower Bound}\label{sec:lower}

We now complete the proof of Theorem~\ref{thm:main}.
Lemma~\ref{lem:decoder} shows that an $\varepsilon$-accurate QLS solver for the instance $(H_X, \ket b)$ gives an algorithm for computing $F(X)$ with success probability strictly greater than $1/2+2\varepsilon$, or equivalently with error probability strictly smaller than $1/2-2\varepsilon$.

To apply the XOR lemma (Lemma~\ref{lem:xor}), we choose its parameter $\delta$ so that the error probability appearing in~\eqref{eq:xor} is exactly $1/2-2\varepsilon$.
Set
\begin{equation}\label{eq:delta}
    \delta
    \coloneqq
    (4\varepsilon)^{2/n}.
\end{equation}
Since $0<\varepsilon\le1/64$, we have $0<4\varepsilon<1$, and therefore $0<\delta<1$, as required in Lemma~\ref{lem:xor}.
Moreover, $\delta^{n/2}=4\varepsilon$, and hence
\begin{equation}
    \frac{1-\delta^{n/2}}{2}
    =
    \frac12-2\varepsilon.
\end{equation}
Thus the error probability achieved by the decoder is strictly smaller than the error probability threshold required to apply Lemma~\ref{lem:xor}.

The lower bound in Lemma~\ref{lem:xor} (i.e., right-hand side of inequality in~\eqref{eq:xor}) contains an additional multiplicative factor $\delta$.
We therefore need to show that our choice of $n$ prevents $\delta$ from becoming arbitrarily small.
From~\eqref{eq:delta}, we have $\log\delta=-2(\log(1/\varepsilon)-\log4)/n>-2\log(1/\varepsilon)/n$.
The lower bound on $n$ in~\eqref{eq:n-main-bounds}, together with the fact that $R\ge16$, gives
$n>(R/8)\log(1/\varepsilon)\ge2\log(1/\varepsilon)$.
Therefore $\log\delta>-1$, and hence
\begin{equation}\label{eq:delta-constant}
    \delta>e^{-1}.
\end{equation}
Thus the factor $\delta$ in the XOR lemma remains bounded below by an absolute positive constant, independently of $n$, $m$, $R$, and $\varepsilon$.

We can now apply the XOR lemma to the function $g$ with $k=n$.
Using the adversary lower bound for a single \textsc{Position-Parity} instance from Lemma~\ref{lem:colored}, we obtain
\begin{align}
    Q_{\frac12-2\varepsilon}(F)
    &\ge
    \frac{n\delta}{8}\Adv^\pm(g)
    \nonumber\\
    &\ge
    \frac{n\delta}{8}\sqrt{\frac{m}{2}}
    \nonumber\\
    &>
    \frac{n}{8e}\sqrt{\frac{m}{2}},
    \label{eq:boolean-lb}
\end{align}
where the last inequality follows from~\eqref{eq:delta-constant}.
In particular, we have
\begin{equation}\label{eq:boolean-lb-asymptotic}
    Q_{\frac12-2\varepsilon}(F)
    =
    \Omega(n\sqrt{m}).
\end{equation}

We next transfer this Boolean query lower bound to sparse-matrix queries to the QLS instance $(H_X, \ket b)$.
Suppose that a QLS algorithm prepares an $\varepsilon$-accurate solution of the instance $(H_X,\ket b)$ using $T$ total queries to $O_{H_X}^{\mathrm{pos}}$, $O_{H_X}^{\mathrm{val}}$, and their inverses.
Lemma~\ref{lem:sparse} shows that each query to $O_{H_X}^{\mathrm{pos}}$ or its inverse can be implemented without making any query to the oracle $O_X$.
The same lemma shows that each query to $O_{H_X}^{\mathrm{val}}$ or its inverse can be implemented using at most two queries to $O_X$.
Therefore the entire QLS algorithm can be simulated using at most $2T$ queries to $O_X$.

After simulating the QLS algorithm, we apply the decoder from Lemma~\ref{lem:decoder}.
The decoder uses no additional oracle queries and computes $F(X)$ with error probability strictly smaller than $1/2-2\varepsilon$.
Thus the simulation gives a Boolean query algorithm for $F$ using at most $2T$ queries.
Combining this fact with~\eqref{eq:boolean-lb-asymptotic} gives
\begin{equation}\label{eq:T-nm}
T = \Omega(n\sqrt{m}).
\end{equation}

It remains only to relate the auxiliary parameters $R$, $m$, and $n$ to the QLS parameters $\kappa$, $s$, and~$\varepsilon$.

We first choose $R$.
For $\kappa\ge31$, let $R$ be the largest power of two satisfying
\begin{equation}\label{eq:R-choice}
    R \le \frac{\kappa+1}{2}.
\end{equation}
Because $(\kappa+1)/2\ge16$, this choice guarantees $R\ge16$, so all estimates established in Section~\ref{sec:linear} apply.
Moreover,~\eqref{eq:H-kappa} gives $\kappa(H_X)=2R-1\le\kappa$, so the constructed matrix satisfies the required condition-number promise.

We also need a lower bound on $R$ in terms of $\kappa$.
Since $R$ is the largest power of two satisfying~\eqref{eq:R-choice}, the next power of two does not satisfy it.
Thus $2R>(\kappa+1)/2$, which implies
\begin{equation}\label{eq:R-kappa-lower}
    R>\frac{\kappa+1}{4}>\frac{\kappa}{4}.
\end{equation}
Therefore $R=\Theta(\kappa)$ within the constructed family.

We choose $m$ similarly in terms of the required sparsity parameter $s$.
For $s\ge5$, let $m$ be the largest power of two satisfying
\begin{equation}\label{eq:m-choice}
    m\le s-1.
\end{equation}
Because $s\ge5$, this gives $m\ge4$, as required in the construction of the \textsc{Position-Parity} problem.
By Lemma~\ref{lem:sparse}, we have
$\operatorname{sparsity}(H_X)\le m+1\le s$, so the matrix $H_X$ obeys the desired sparsity promise.
Furthermore, since the next power of two is larger than $s-1$, we have
$2m>s-1$, and therefore
$m>(s-1)/2$.
For $s\ge5$, this implies $m>s/4$, so
\begin{equation}\label{eq:sqrt-m-s}
    \sqrt{m}
    >
    \frac{\sqrt{s}}{2}.
\end{equation}
Thus $m=\Theta(s)$ and the factor $\sqrt{m}$ in~\eqref{eq:T-nm} becomes a factor $\Theta(\sqrt{s})$.

Finally, we relate $n$ to $\kappa$ and $\varepsilon$.
The lower bound on $n$ in~\eqref{eq:n-main-bounds} and the relation $R>\kappa/4$ from~\eqref{eq:R-kappa-lower} give
\begin{equation}\label{eq:n-kappa-epsilon}
    n
    >
    \frac{\kappa}{32}
    \log\frac1\varepsilon.
\end{equation}
Thus the number of composed \textsc{Position-Parity} instances satisfies
$n=\Omega(\kappa\log(1/\varepsilon))$.

Substituting~\eqref{eq:sqrt-m-s} and~\eqref{eq:n-kappa-epsilon} into the query lower bound~\eqref{eq:T-nm} gives
\begin{equation}
    T = \Omega(n\sqrt{m}) = \Omega\!\left( \kappa\sqrt{s}\log\frac1\varepsilon \right).
\end{equation}
This proves Theorem~\ref{thm:main}.

We conclude by evaluating the dimension of the constructed QLS instance.
By~\eqref{eq:H-dimension}, the matrix $H_X$ has dimension $N=6mn$.
Using $m\le s$, the upper bound
$n<(R/2)\log(1/\varepsilon)$ from~\eqref{eq:n-main-bounds}, and $R\le(\kappa+1)/2$ from~\eqref{eq:R-choice}, we obtain
\begin{equation}
    N < \frac{3}{2} s(\kappa+1) \log\frac1\varepsilon.
\end{equation}
Hence the dimension of the hard instances grows only polynomially with the parameters appearing in the lower bound.

\section{Discussion}\label{sec:discussion}

In this paper, we proved the optimal joint sparse-access lower bound $\Omega(\kappa\sqrt{s}\log(1/\varepsilon))$ for the QLS problem, thereby settling the conjectured dependence on the condition number $\kappa$, sparsity $s$, and target error $\varepsilon$. 
The three factors arise from two parts of the construction: a single \textsc{Position-Parity} instance has query complexity $\Omega(\sqrt{m})$, with $m=\Theta(s)$, while the linear-system construction combines $n=\Theta(\kappa\log(1/\varepsilon))$ independent such instances and encodes their XOR into the QLS solution state. 
By Lemma~\ref{lem:decoder}, an $\varepsilon$-accurate QLS solution allows this XOR to be computed with error probability smaller than $1/2-2\varepsilon$, and the XOR lemma then implies that doing so requires $\Omega(n\sqrt{m})$ Boolean queries. 
This realizes the direction of composing Boolean functions identified by the authors of Ref.~\cite{Mori2026} in their discussion as a possible route to a joint lower bound in $\kappa$, $s$, and $\varepsilon$.
As an open question, the hard matrix $H_X$ constructed here is real and symmetric but indefinite, and it would be interesting to determine whether the same joint lower bound continues to hold under additional structural assumptions, such as positive definiteness.

\section{Acknowledgments}
The author is grateful to Alexander M. Dalzell for helpful discussions.
The author acknowledges the use of ChatGPT 5.6 for help with writing and reviewing the manuscript, including identifying a few technical errors and inconsistencies in earlier versions.
This work was supported by the U.S. Department of Energy, Office of Science, National Quantum Information Science Research Centers, Co-design Center for Quantum Advantage (C2QA) under contract number DE-SC0012704, and by the U.S. National Science Foundation National Quantum Virtual Laboratory (NQVL), Erasure Qubits and Dynamic Circuits for Quantum Advantage (ERASE), under grant number OSI-2435244.

\bibliographystyle{unsrt}
\bibliography{ref}

\end{document}